\documentclass[runningheads]{llncs}
\usepackage{fullpage}
\usepackage{hyperref}
\usepackage[T1]{fontenc}
\usepackage{algorithm} %% For the floating Algorithm environment.
\usepackage{jybyPython} %% For displaying (and including) Python code.
\providecommand{\etal}{{\em et al.}}
\providecommand{\idtt}[1]{\ensuremath{\mathtt{#1}}}

\usepackage{graphicx}
\usepackage{version}
\excludeversion{INUTILE}
\includeversion{LONG}\excludeversion{SHORT}\includeversion{PEDAGOGICALEXAMPLE}
\excludeversion{TODO}
\excludeversion{PYTHON}
\excludeversion{CONDITIONAL}

\begin{document}
\title{Indexed Dynamic Programming to boost \\ Edit Distance and LCSS Computation\thanks{Supported by project Fondecyt Regular no 1170366 from Conicyt.}}
\titlerunning{Index to Boost Edit Distance and LCSS computation}
\author{J\'er\'emy Barbay, Andr{\'e}s Olivares}
\authorrunning{J. Barbay, A. Olivares}   % abbreviated author list (for running head)
%
%%%% modified list of authors for the TOC (add the affiliations)
\tocauthor{J\'er\'emy Barbay (Universidad de Chile), Andr{\'e}s Olivares (Universidad de Chile)}
\institute{
  Departamento de Ciencias de la Computaci{\'o}n, 
  University of Chile, 
\email{jeremy@barbay.cl, aolivare@dcc.uchile.cl}
}

\maketitle              % typeset the header of the contribution
%
% The abstract should briefly summarize the contents of the paper in 15--250 words.
\begin{abstract}
There are efficient dynamic programming solutions to the computation of the Edit Distance from $S\in[1..\sigma]^n$ to $T\in[1..\sigma]^m$, for many natural subsets of edit operations, typically in time within $O(nm)$ in the worst-case over strings of respective lengths $n$ and $m$ (which is likely to be optimal), and in time within $O(n{+}m)$ in some special cases (e.g. disjoint alphabets).  We describe how indexing the strings (in linear time), and using such an index to refine the recurrence formulas underlying the dynamic programs, yield faster algorithms in a variety of models, on a continuum of classes of instances of intermediate difficulty between the worst and the best case, thus refining the analysis beyond the worst case analysis.  As a side result, we describe similar properties for the computation of the Longest Common Sub Sequence $\idtt{LCSS}(S,T)$ between $S$ and $T$, since it is a particular case of Edit Distance, and we discuss the application of similar algorithmic and analysis techniques for other dynamic programming solutions.  More formally, we propose a parameterized analysis of the computational complexity of the Edit Distance for various set of operators and of the Longest Common Sub Sequence in function of the area of the dynamic program matrix relevant to the computation.
\end{abstract}

\providecommand{\alphabetSize}{\sigma}

\section{Introduction}

%% CONTEXT 
Given a set of edition operators on strings, a source string $S\in[1..\sigma]^n$ and a target string $T\in[1..\sigma]^m$ of respective lengths $n$ and $m$  on the alphabet $[1..\alphabetSize]$, the {\sc Edit Distance} is the minimum number of such operations required to transform the string~$S$ into the string~$T$. Introduced in 1974 by Wagner and Fischer~\cite{1974-JACM-TheStringToStringCorrectionProblem-WagnerFischer}, such computation is a fundamental problem in Computer Science, with a wide range of applications, from text processing and information retrieval to computational biology.
The typical edition distance between two strings is defined by the minimum number of \texttt{insertions}, \texttt{deletions} (in both cases, of a character at an arbitrary position of $S$) and \texttt{replacement} (of one character of $S$ by some other) needed to transform the string $S$ into $T$. Many generalizations have been defined in the literature, including weighted costs for the edit operations, and different sets of edit operations -- the standard set is \{\texttt{insertion}, \texttt{deletion}, \texttt{replacement}\}.

%% Traditional solutions
% \paragraph{Traditional solutions:}
Each distinct set of correction operators yields a distinct correction distance on strings  (see Figure \ref{fig:overview} for a summary).
\begin{figure}[t]
\centering
\begin{SHORT}
\begin{tabular}{c|c|c|c}
            & $(n,m)$-Worst                                                                       & \multicolumn{2}{c}{Finer Results}                                                                                 \\ 
  Operators & Case Complexity                                                                     & Distance                                                                                   & Parikh vectors       \\ 
\hline
% Delete    & $O(m)$    \cite{2013-Master-TheBinaryStringtoStringCorrectionProblem-Spreen}        &                                                                                                                   \\
% Insert    & $O(n)$     \cite{2013-Master-TheBinaryStringtoStringCorrectionProblem-Spreen}       &                                                                                                                   \\
% Replace   & $O(n)$           \cite{2013-Master-TheBinaryStringtoStringCorrectionProblem-Spreen} &                                                                                                                   \\
Swap        & $O(n^2)$         \cite{2013-Master-TheBinaryStringtoStringCorrectionProblem-Spreen} &  $O(n(1+\lg d/n)\lg n)$ & DNA \\
            % &                                                                                     &$O(n(1+\lg d/n)\lg\lg\sigma)$        &                      \\

\hline
Delete, Insert                & $O(nm)$          \cite{2000-SPIRE-ASurveyOfLongestCommonSubsequenceAlgorithms-BergrothHakonenRaita}  & $O(d^2)$ & $O(\sum n_\alpha m_\alpha\lg nm)$\\
Delete, Replace               & $O(nm)$          \cite{1975-STOC-OnTheComplexityOfTheExtendedStringToStringCorrectionProblem-Wagner}& $O(d^2)$ & $O(\sum n_\alpha m_\alpha\lg nm)$\\
% = Insert, Replace               & $O(nm)$         \cite{1975-STOC-OnTheComplexityOfTheExtendedStringToStringCorrectionProblem-Wagner} & $O(d^2)$ & \\
Delete, Swap                & NP-complete   \cite{1974-JACM-TheStringToStringCorrectionProblem-WagnerFischer}                      & $O({1.6181}^d m)$  \cite{2011-DO-ChargeAndReduceAFixedParameterAlgorithmForStringToStringCorrection-AbukhzamFernauLangstonLeeculturaStege} & $O(\sigma^2nm\gamma^{\sigma-1})$ \cite{2015-SPIRE-AdaptiveComputationOfTheSwapInsertCorrectionDistance-BarbayPerez} \\
% = Insert, Swap                & NP-complete   \cite{1974-JACM-TheStringToStringCorrectionProblem-WagnerFischer}                      & $O({1.6181}^d m)$  \cite{2011-DO-ChargeAndReduceAFixedParameterAlgorithmForStringToStringCorrection-AbukhzamFernauLangstonLeeculturaStege} & $O(\sigma^2nm\gamma^{\sigma-1})$ \cite{2015-SPIRE-AdaptiveComputationOfTheSwapInsertCorrectionDistance-BarbayPerez} \\
Replace, Swap                 & $O(nm)$          \cite{1975-STOC-OnTheComplexityOfTheExtendedStringToStringCorrectionProblem-Wagner} & $O(d^2)$&  $O(\sum n_\alpha m_\alpha\lg nm)$\\
\hline
Delete, Insert, Replace       & $O(nm)$          \cite{2000-SPIRE-ASurveyOfLongestCommonSubsequenceAlgorithms-BergrothHakonenRaita}  & $O(d^2)$&  $O(\sum n_\alpha m_\alpha\lg nm)$ \\
Delete, Insert, Swap          & $O(nm)$          \cite{1974-JACM-TheStringToStringCorrectionProblem-WagnerFischer}                   & $O(d^2)$&  \\
Delete, Replace, Swap         & $O(nm)$          \cite{1974-JACM-TheStringToStringCorrectionProblem-WagnerFischer}                   & $O(d^2)$&  \\
Insert, Replace, Swap         & $O(nm)$          \cite{1974-JACM-TheStringToStringCorrectionProblem-WagnerFischer}                   & $O(d^2)$& \\
\hline
Delete, Insert,  & $O(nm)$          \cite{2000-SPIRE-ASurveyOfLongestCommonSubsequenceAlgorithms-BergrothHakonenRaita}  & $O(d^2)$& \\
Replace, Swap & & & \\
\end{tabular}
\end{SHORT}
\begin{LONG}
\begin{tabular}{c|c|c|c}
            & $(n,m)$-Worst                                                                       & \multicolumn{2}{c}{Finer Results}                                                                                 \\ 
  Operators & Case Complexity                                                                     & Distance                                                                                   & Parikh vectors       \\ 
\hline
Delete    & $O(m)$    \cite{2013-Master-TheBinaryStringtoStringCorrectionProblem-Spreen}        &                                                                                                                   \\
Insert    & $O(n)$     \cite{2013-Master-TheBinaryStringtoStringCorrectionProblem-Spreen}       &                                                                                                                   \\
Replace   & $O(n)$           \cite{2013-Master-TheBinaryStringtoStringCorrectionProblem-Spreen} &                                                                                                                   \\
Swap        & $O(n^2)$         \cite{2013-Master-TheBinaryStringtoStringCorrectionProblem-Spreen} & $O(n(1+\lg d/n)\lg\lg\sigma)$ & DNA \\
            &                                                                                     & $O(n(1+\lg d/n)\lg n)$        &                      \\

\hline
Delete, Insert                & $O(nm)$          \cite{2000-SPIRE-ASurveyOfLongestCommonSubsequenceAlgorithms-BergrothHakonenRaita}  & $O(d^2)$ &                                                                                                                                                                                                                                                                  \\
Delete, Replace               & $O(nm)$          \cite{1975-STOC-OnTheComplexityOfTheExtendedStringToStringCorrectionProblem-Wagner}& $O(d^2)$ &                                                                                                                                                                                                                                                                  \\
= Insert, Replace               & $O(nm)$         \cite{1975-STOC-OnTheComplexityOfTheExtendedStringToStringCorrectionProblem-Wagner} & $O(d^2)$ &                                                                                                                                                                                                                                                                  \\
Delete, Swap                & NP-complete   \cite{1974-JACM-TheStringToStringCorrectionProblem-WagnerFischer}                      & $O({1.6181}^d m)$  \cite{2011-DO-ChargeAndReduceAFixedParameterAlgorithmForStringToStringCorrection-AbukhzamFernauLangstonLeeculturaStege} & $O(\sigma^2nm\gamma^{\sigma-1})$ \cite{2015-SPIRE-AdaptiveComputationOfTheSwapInsertCorrectionDistance-BarbayPerez} \\
= Insert, Swap                & NP-complete   \cite{1974-JACM-TheStringToStringCorrectionProblem-WagnerFischer}                      & $O({1.6181}^d m)$  \cite{2011-DO-ChargeAndReduceAFixedParameterAlgorithmForStringToStringCorrection-AbukhzamFernauLangstonLeeculturaStege} & $O(\sigma^2nm\gamma^{\sigma-1})$ \cite{2015-SPIRE-AdaptiveComputationOfTheSwapInsertCorrectionDistance-BarbayPerez} \\
Replace, Swap                 & $O(nm)$          \cite{1975-STOC-OnTheComplexityOfTheExtendedStringToStringCorrectionProblem-Wagner} & $O(d^2)$& \\
\hline
Delete, Insert, Replace       & $O(nm)$          \cite{2000-SPIRE-ASurveyOfLongestCommonSubsequenceAlgorithms-BergrothHakonenRaita}  & $O(d^2)$& \\
Delete, Insert, Swap          & $O(nm)$          \cite{1974-JACM-TheStringToStringCorrectionProblem-WagnerFischer}                   & $O(d^2)$& \\
Delete, Replace, Swap         & $O(nm)$          \cite{1974-JACM-TheStringToStringCorrectionProblem-WagnerFischer}                   & $O(d^2)$& \\
Insert, Replace, Swap         & $O(nm)$          \cite{1974-JACM-TheStringToStringCorrectionProblem-WagnerFischer}                   & $O(d^2)$& \\
\hline
Delete, Insert,  & $O(nm)$          \cite{2000-SPIRE-ASurveyOfLongestCommonSubsequenceAlgorithms-BergrothHakonenRaita}  & $O(d^2)$& \\
Replace, Swap & & & \\
\end{tabular}
\end{LONG}
\begin{minipage}{.9\linewidth}
\caption{Summary of  results for various combinations of operators from the basic set \{\texttt{Insert}, \texttt{Delete}, \texttt{Replace}, \texttt{Swap}\}.  The column labeled ``$(n,m)$-Worst Case Complexity'' presents results in the worst case over instances of fixed sizes $n$ and $m$, while the columns labeled ``Finer Results'' present results where the analysis was refined by various parameters: the distance $d$, the size $\sigma$ of the alphabet, and some form of imbalance $\gamma = \max_{\alpha\in[1..\sigma]}\min\{n_\alpha,m_\alpha{-}n_\alpha\}$ between the \texttt{Parikh vectors} of $S$ and $T$. For brevity, the only distance based on a single operator presented is the \textsc{Swap Edit Distance}, as the computation of the others is always linear in the size of the input. }
\end{minipage}
\label{fig:overview}
\end{figure}
For instance, Wagner and Fischer~\cite{1974-JACM-TheStringToStringCorrectionProblem-WagnerFischer} showed that for the three following operations, the \texttt{insertion} of a symbol at some arbitrary position, the \texttt{deletion} of a symbol at some arbitrary position, and the \texttt{replacement} of a symbol at some arbitrary position, the \textsc{Edit Distance} can be computed in time within $O(nm)$ and space within $O(n{+}m)$ using traditional dynamic programming techniques. 
As another variant of interest, Wagner and Lowrance~\cite{1975-JACM-AnExtensionOfTheStringToStringCorrectionProblem-WagnerLowrance} introduced the \texttt{Swap} operator (\texttt{S}), which exchanges the positions of two contiguous symbols.
When considering only the \texttt{Swap} operator, one basically searches for the permutation transforming the source string $S$ into the target string $T$: some adaptive sorting technique yields a minor improvement on the computation of the \textsc{Swap Edit Distance} (see appendix~\ref{sec:swap}).
For two of the newly defined distances, the \textsc{Insert Swap Edit Distance} and the \textsc{Delete Swap Edit Distance} (equivalent by symmetry), the best known algorithms take time exponential in the input size~\cite{2015-SPIRE-AdaptiveComputationOfTheSwapInsertCorrectionDistance-BarbayPerez,2018-TALG-AdaptiveComputationOfTheSwapInsertCorrectionDistance-BarbayPerez}, which is likely to be optimal~\cite{1974-JACM-TheStringToStringCorrectionProblem-WagnerFischer}.
The \textsc{Edit Distance} itself is linked to many other problems: for instance, given the two same strings $S\in[1..\alphabetSize]^n$ and $T\in[1..\alphabetSize]^m$, the computation of the \textsc{Longest Common Sub-Sequence} (LCSS) $L$ between $S$ and $T$ is equivalent to the computation of the \textsc{Delete Insert Edit Distance} $d$, as the symbols deleted from $S$ and inserted from $T$ in order to ``edit'' $S$ into $T$ are exactly the same as the symbols deleted from $S$ and $T$ in order to produce $L$. Hence, the LCSS between $S$ and $T$ can be computed in time within $O(nm)$ and space within $O(n{+}m)$ using traditional dynamic programming techniques.

Most of these computational complexities are likely to be optimal in the worst case over instances of size $(n,m)$: the algorithms computing the three basic distances
\begin{LONG}
(\textsc{Insert Edit Distance}, \textsc{Delete Edit Distance} and \textsc{Replace Edit Distance})
\end{LONG}
in linear time are optimal as any algorithm must read the whole strings; the \textsc{Insert Swap Edit Distance} and its symmetric the \textsc{Delete Swap Edit Distance} are NP-hard to compute~\cite{1975-STOC-OnTheComplexityOfTheExtendedStringToStringCorrectionProblem-Wagner}; and in 2015 Backurs and Indyk~\cite{2015-STOC-EditDistanceCannotBeComputedInStronglySubquadraticTimeUnlessSETHIsFalse-BackursIndyk} showed that the $O(n^2)$ upper bound for the computation of the \textsc{Delete Insert  Replace Edit Distance} is optimal unless the \emph{Strong Exponential Time Hypothesis} (SETH) is false.  

%% Adaptive results
% \paragraph{Using a Rank and Select index:}
More recently, Barbay and P{\'{e}}rez{-}Lantero~\cite{2015-SPIRE-AdaptiveComputationOfTheSwapInsertCorrectionDistance-BarbayPerez,2018-TALG-AdaptiveComputationOfTheSwapInsertCorrectionDistance-BarbayPerez}, complementing Meister's previous results~\cite{2015-TCS-UsingSwapsAndDeletesToMakeStringsMatch-Meister} by the use of an index supporting the operators \texttt{rank} and \texttt{select} on strings, described an algorithm computing this distance in time within $O(\sigma^2nm\gamma^{\sigma-1})$ in the worst case over instances where $\sigma,n,m$ and $\gamma$ are fixed, where $\gamma= \max_{\alpha\in[1..\sigma]}\min\{n_\alpha,m_\alpha{-}n_\alpha\}$ measures a form of imbalance between the frequency distributions of each string.

%% Hypothesis
\paragraph{Hypothesis:}
Given this situation,  is it possible to \textbf{take advantage of} indexing techniques supporting \textbf{rank and select} in order \textbf{to speed up the computation of} other \textbf{edit distances}?
Can a similar analysis to that of Barbay and P\'erez-Lantero~\cite{2015-SPIRE-AdaptiveComputationOfTheSwapInsertCorrectionDistance-BarbayPerez,2018-TALG-AdaptiveComputationOfTheSwapInsertCorrectionDistance-BarbayPerez} be applied to other edit distances? 
\textbf{Are there instances for which the edit distance is easier to compute}, and \textbf{do such instances occur in real applications} of the computation of the edit distance?

%% Results
\paragraph{Our Results:}
We answer all those questions positively, and describe general techniques to refine the analysis of dynamic programs beyond the traditional analysis in the worst case over input of fixed size. More specifically, we analyze the computational cost of four \textsc{Edit Distances} using various \texttt{rank} and \texttt{select} text indices, in function of the \texttt{Parikh vector}~\cite{2017-wikipedia-Parikhtheorem} of the source $S$ and target $T$ strings. As a side result, this yields similar properties for the computation of the \textsc{Longest Common Sub Sequence} $\idtt{LCSS}(S,T)$ between $S$ and $T$\begin{LONG}, as it can be deduced from the \textsc{Delete Insert Edit Distance} ($\idtt{LCSS}(S,T)=|S|+|T|-2d_{DI}(S,T)$)\end{LONG}, and definitions and techniques which can be applied to other dynamic programs.
After defining formally the notion of \texttt{Parikh's vector} and various index data structures supporting \texttt{rank} and \texttt{select} on strings in Section~\ref{sec:background}, we describe the algorithms taking advantage of such techniques in Section~\ref{sec:algorithms}: for the \textsc{Longest Common Sub Sequence} and \textsc{Delete-Insert Edit Distance} (Section~\ref{sec:DIUpper}), the \textsc{Delete Insert Replace Edit Distance} (Section~\ref{sec:DIRUpper}), and finally for the \textsc{Delete-Replace Edit Distance} and its dual the \textsc{Insert-Replace Edit Distance} (Section~\ref{sec:DRUpper}).
We describe some preliminary experiments and their results, which seem to indicate that those instances are not totally artificial and occur naturally in practical applications in Section~\ref{sec:experiments}.
We conclude in Section~\ref{sec:discussion} with a discussion of other potential refinement of the analysis, and the extension of our results to other \textsc{Edit Distances}.

\section{Preliminaries}
\label{sec:background}

Before describing our proposed algorithms to compute various \textsc{Edit Distances}, we describe formally in Section~\ref{sec:parikhVector} the notion of \texttt{Parikh vector} which is essential to our analysis technique; and in Section~\ref{sec:rankSelectImplementations} two key implementations of indices supporting the \texttt{rank} and \texttt{select} operators on strings.

\subsection{Parikh vector}
\label{sec:parikhVector}

Given positive integers $\sigma$ and $n$, a string $S\in[1..\sigma]^n$, and the integers $n_1,\ldots,n_\sigma$ such that $n_\alpha$ denotes the number of occurrences of the letter $\alpha\in[1..\sigma]$ in the string $S$, the \texttt{Parikh vector} of $S$ is defined~\cite{2017-wikipedia-Parikhtheorem} as
$p(S)=(n_1,\ldots,n_\sigma).$

Barbay and P{\'{e}}rez{-}Lantero~\cite{2015-SPIRE-AdaptiveComputationOfTheSwapInsertCorrectionDistance-BarbayPerez} refined the analysis of the \textsc{Insert Swap Edit Distance} from a string $S\in[1..\sigma]^n$ to a string $T\in[1..\sigma]^m$ via a function of the \texttt{Parikh vectors} $(n_1,\ldots,n_\sigma)$ of $S$ and $(m_1,\ldots,m_\sigma)$ of $T$, the local imbalance $\gamma_{\alpha}=\min\{n_{\alpha},m_{\alpha}-n_{\alpha}\}$ for each symbol $\alpha\in[1..\sigma]$, projected to a global measure of imbalance, $\gamma=\max_{\alpha\in[1..\sigma]} \gamma_\alpha$.
In the worst case among instances of fixed \texttt{Parikh vector}, they describe an algorithm to compute the \textsc{Insert Swap Edit Distance} in time within 
\begin{LONG}
$$
 O\left(
  dn 
  +d^2n
  \cdot 
  \left(\sum_{\alpha=1}^\sigma(m_{\alpha}-\gamma_{\alpha})\right)
  \cdot
    \prod_{\alpha\in\overline{[1..\sigma]}}(\gamma_\alpha+1)
  \right),
$$
where 
$\overline{[1..\sigma]}=\{\alpha\in[1..\sigma]: \gamma_{\alpha}>0\}$ if $\Pi_{\alpha\in[1..\sigma]}\gamma_\alpha=0$, and
$\overline{[1..\sigma]}=[1..\sigma]\setminus \{\arg\min_{\alpha\in[1..\sigma]}\gamma_{\alpha}\}$ otherwise.
This formula simplifies to within
\end{LONG}
$O(\sigma^2nm\gamma^{\sigma-1})$ in the worst case over instances where $\sigma,n,m$ and $\gamma$ are fixed.

Such a vector is essential to the fine analysis of dynamic programs for computing \textsc{Edit Distances} when using operators whose running time depends on the number of occurrence of each symbol, such as for the \texttt{rank} and \texttt{select} operators described in the next section.

\subsection{Rank and Select in Strings}
\label{sec:rankSelectImplementations}

\begin{LONG}
For every string $X\in\{S,L\}$ and integer $i\in[1..|X|]$, $X[i]$ denotes the $i$-th symbol of $X$ from left to right. For every pair of integers $i,j\in[1..|X|]$ such that $i\le j$, $X[i..j]$ denotes the substring of $X$ from the $i$-th symbol to the $j$-th symbol, and for every pair of integers $i,j\in[1..|X|]$ such that $j<i$, $X[i..j]$ denotes the empty string.
\end{LONG}

Given a symbol $\alpha\in[1..\sigma]$, an integer $i\in[1..|X|]$ and an integer $k>0$, the operation $rank(X,i,\alpha)$ denotes the number of occurrences of the symbol $\alpha$ in the substring $X[1..i]$, and the operation $select(X,k,\alpha)$ denotes the value $j\in[1..|X|]$ such that the $k$-th occurrence of $\alpha$ in $X$ is precisely at position $j$, if $j$ exists. If $j$ does not exist, then $select(X,k,\alpha)$ is $null$.

A simple way to support the \texttt{rank} and \texttt{select} operators in reasonably good time consists in, for each symbol $\alpha\in[1..\sigma]$, listing all the occurrences of $\alpha$ in a sorted array (called a ``Posting List''~\cite{1999-BOOK-ManagingGigabytes-WittenMoffatBell}): supporting the \texttt{select} operator reduces to a simple access to the sorted array corresponding to the symbol $\alpha$; while supporting the \texttt{rank} operator reduces to a \textsc{Sorted Search} in the same array, which can be simply implemented by a \texttt{Binary Search}, or more efficiently in practice by a \texttt{Doubling Search}~\cite{1976-IPL-AnAlmostOptimalAlgorithmForUnboundedSearching-BentleyYao} in time within $O(q_\alpha\lg(n_\alpha/q_\alpha))$ when supporting $q_\alpha$ monotone queries in a posting list of size $n_\alpha$ (for a given symbol $\alpha\in[1..\sigma]$.
\begin{LONG}
\begin{lemma} %[Bentley and Yao~\cite{1976-IPL-AnAlmostOptimalAlgorithmForUnboundedSearching-BentleyYao}]
Given a string $S\in[1..\sigma]^n$ of \texttt{Parikh vector} $(n_1,\ldots,n_\alpha)$, there exists an index using $n+\sigma$ machine words, which can be computed in time linear in the size $n$ of $S$ in order to support the operators \texttt{rank} and \texttt{select} in time within $O(q_\alpha\lg(n_\alpha/q_\alpha))$ in the comparison based decision tree model, when $q_\alpha$ of those queries concern the symbol $\alpha\in[1..\sigma]$.
\end{lemma}
\end{LONG}

Golynski~\etal~\cite{2006-SODA-RankSelectOperationsOnLargeAlphabetsAToolForText-GolynskiMunroRao} described a more sophisticated (but asymptotically more efficient) way to support the \texttt{rank} and \texttt{select} operators in the RAM model, via a clever reduction to $Y$-Fast Trees on permutations supporting the operators in time within $O(\lg\lg\sigma)$.
\begin{LONG}
Barbay et al.~\cite{2014-Algorithmica-EfficientFullyCompressedSequenceRepresentations-BarbayClaudeGagieNavarroNekrich} showed that it can be done on a compressed representation of the text.
\begin{lemma} %[Barbay et al.~\cite{2014-Algorithmica-EfficientFullyCompressedSequenceRepresentations-BarbayClaudeGagieNavarroNekrich}]
Given a string $S\in[1..\sigma]^n$, there exists an index using space within $o(n\lg\sigma)$ bits, which can be computed in time linear in the size $n$ of $S$ in order to support the operators \texttt{rank} and \texttt{select} in time within $O(\lg\lg\sigma)$ in the RAM model.
\end{lemma}
\end{LONG}
We describe how to use those techniques to speed up the computation of various \textsc{Edit Distances} in the following sections.

\section{Adaptive Dynamic Programs}
\label{sec:algorithms}

For each of the problems considered, we describe how to compute a subset of the values computed by classical dynamic programs.  We start with the computation of the \textsc{Longest Common Sub Sequence} (LCSS) and the \textsc{Delete Insert (DI) Edit Distance} (Section~\ref{sec:DIUpper}) because it is the simplest; extend its results to the computation of the \textsc{Levenshtein Edit Distance} (Section~\ref{sec:DIRUpper}); and project those to the computation of the \textsc{Delete Replace (DR) Edit Distance} and its symmetric \textsc{Insert Replace (IR) Edit Distance} (Section~\ref{sec:DRUpper}). 

\subsection{LCSS and DI-Edit Distance}
\label{sec:DIUpper}

The \textsc{Delete Insert Edit Distance} is a classical problem in Stringology~\cite{2000-SPIRE-ASurveyOfLongestCommonSubsequenceAlgorithms-BergrothHakonenRaita}, if only as a variant of the \textsc{Longest Common Sub Sequence} problem.  It is classically computed using dynamic programming: we describe the classical solution first, which we then refine\begin{PEDAGOGICALEXAMPLE} in a simplistic way, as a pedagogical introduction to a more sophisticated refinement\end{PEDAGOGICALEXAMPLE}.

\subsubsection{Classical solution:}
\label{sec:DeleteInsertClassical}

Given two strings $S\in[1..\sigma]^n$ and $T\in[1..\sigma]^m$, we note $d_{DI}(n,m)$ the \textsc{Delete Insert Edit Distance} from $S$ to $T$. If the last symbols of $S$ and $T$ match, the edit distance is the same as the edit distance between the prefixes of respective lengths $n-1$ and $m-1$ of $S$ and $T$. Otherwise, the edit distance is the minimum between the edit distance when inserting a copy of the last symbol of $T$ in $S$ (i.e. deleting this symbol in $T$) and the edit distance when deleting the mismatching symbol in $T$. More formally:

\begin{eqnarray*}
  d_{DI}(S[1..n],T[1..m]) 
&=& \left\{
    \begin{array}{ll}
      n \mbox{ if $m==0$;}\\
      m \mbox{ if $n==0$;}\\
      d_{DI}(S[1..n-1],T[1..m-1])  \mbox{ if $S[n]==T[m]$; and } \\
      1 + \min\left\{
      \begin{array}{l}
  d_{DI}(S[1..n-1],T[1..m]), \\
  d_{DI}(S[1..n],T[1..m-1]) 
      \end{array}
\right\} \mbox{ otherwise.}
    \end{array}
\right.
\end{eqnarray*}

This recursive definition directly yields an algorithm to compute the \textsc{Delete Insert Edit Distance} from $S$ to $T$ in time within $O(nm)$ and space within $O(n+m)$. We describe in the next section a technique taking advantage of the discrepancies between the \texttt{Parikh vectors} of $S$ and $T$.

\begin{PEDAGOGICALEXAMPLE}
\subsubsection{A Pedagogical Example:}
\label{sec:InsertDeletePedagogicalExample}

Given two strings $S\in[1..\sigma]^n$ and $T\in[1..\sigma]^m$, for each symbol $\alpha\in[1..\sigma]$, let's note $n_\alpha$ and $m_\alpha$ the number of occurrences of $\alpha$ respectively in $S$ and $T$. Assembled in a vector, those form the \texttt{Parikh vectors} $(n_\alpha)_{\alpha\in[1..\sigma]}$ for $S$ and $(m_\alpha)_{\alpha\in[1..\sigma]}$ for $T$.
\begin{LONG}
Barbay and P{\'{e}}rez{-}Lantero~\cite{2015-SPIRE-AdaptiveComputationOfTheSwapInsertCorrectionDistance-BarbayPerez} described an algorithm to compute the \textsc{Insert Swap Edit Distance} which complexity is expressed in function of how the \texttt{Parikh vectors} of $S$ and $T$ differ.  Likewise, we describe how those affect the difficulty of computing the \textsc{Delete Insert  Edit Distance} from $S$ to $T$.
\end{LONG}

Consider in Figure~\ref{fig:InsertDeleteEditDistanceDynamicProgram} the graphical representation $M_{DI}(S,T)$ of the dynamic program computing the \textsc{Delete Insert  Edit Distance} from $S$ to $T$, following the dynamic program described in the previous section. For general $i\in[1..n]$ and $j\in[1..m]$, the $i$-th value in the $j$-th row, $a=M_{DI}(S,T)[i,j]=d_{DI}(S[1..i],T[1..j])$ is computed by taking the minimum between $b=M_{DI}(S,T)[i-1,j]=d_{DI}(S[1..i-1],T[1..j])$ and $c=M_{DI}(S,T)[i,j-1]=d_{DI}(S[1..i],T[1..j-1])$, the value directly on the left and directly above it: $a=\min\{b,c\}$.

\begin{figure}
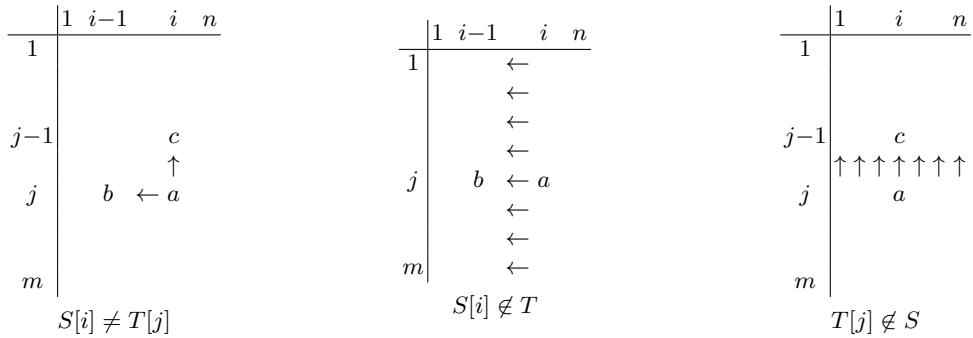

\centering
\begin{minipage}{.3\linewidth}
$$
\begin{array}{c|*{9}{c}}
      & 1        &          & i{-}1    &                          & i        &          &          & n        \\ \hline
  1   &                                                                                                       \\
                                                                                                              \\
                                                                                                              \\
j{-}1 &          &          &          &                          & c        &          &          &          \\
      &          &          &          &                          & \uparrow                                  \\
j     &          &          & b        & \leftarrow               & a        &          &          &          \\
                                                                                                              \\
                                                                                                              \\
 m    &                                                                                                       \\
\end{array}$$
$$S[i]\neq T[j]$$
\end{minipage}
\begin{minipage}{.3\linewidth}
$$\begin{array}{c|*{9}{c}}
      & 1        &          & i{-}1    &                          & i        &          &          & n        \\ \hline
1     &          &          &          & \leftarrow               &          &          &          &          \\
      &          &          &          & \leftarrow               &          &          &          &          \\
      &          &          &          & \leftarrow               &          &          &          &          \\
      &          &          &          & \leftarrow               &          &          &          &          \\
j     &          &          & b        & \leftarrow               & a        &          &          &          \\
      &          &          &          & \leftarrow               &          &          &          &          \\
      &          &          &          & \leftarrow               &          &          &          &          \\
 m    &          &          &          & \leftarrow               &          &          &          &          \\
\end{array}$$
$$S[i]\not\in T$$
\end{minipage}
\begin{minipage}{.3\linewidth}
$$\begin{array}{c|*{9}{c}}
      & 1        &          &   & i        &          &          & n        \\ \hline
  1   &                                                                            \\
                                                                                   \\
                                                                                   \\
j{-}1 &          &          &          & c        &          &          &          \\
      & \uparrow & \uparrow & \uparrow & \uparrow & \uparrow & \uparrow & \uparrow \\
j     &          &          &          & a        &          &          &          \\
                                                                                                              \\
                                                                                                              \\
 m    &                                                                                                       \\
\end{array}$$
$$T[j]\not\in S$$
\end{minipage}
\caption{A graphical representation of the dynamic program computing the  \textsc{Delete Insert  Edit Distance} from $S$ to $T$ in the general case ($S[i]\neq T[j]$) and in the particular case where the symbol at position $i$ in $S$ does not occur in $T$ ($S[i]\not\in T$), and where the symbol at position $j$ in $T$ does not occur in $S$ ($T[j]\not\in S$).}
\label{fig:InsertDeleteEditDistanceDynamicProgram}
\end{figure}

Consider a particular position $i\in[1..n]$ in $S$ such that the symbol $S[i]$ at this position does not occur in $T$ (i.e. $S[i]\not\in T$): this symbol will be deleted in any edition of $S$ into $T$, so that each value in the column $i$ can be obtained by merely duplicating the corresponding one in the column $i{-}1$.
Similarly, consider a particular position $j\in[1..m]$ in $T$ such that the symbol $T[j]$ at this position  does not occur in $S$ (i.e. $T[i]\not\in S$): this symbol will be inserted in any edition of $S$ into $T$, so that each value in the row $j$ can be obtained by merely duplicating the corresponding one in the row $j{-}1$.
The duplication of such columns and row can be simulated in constant time during the execution of the dynamic program, thus reducing the complexity to within $O(n'm'+n+m+\sigma)$ where $n'$ and $m'$ are the lengths of $S$ and $T$ once projected to the intersection of their effective alphabets: $n'=\sum_{\alpha, m_\alpha>0}n_\alpha$ and $m'=\sum_{\alpha, n_\alpha>0}m_\alpha$.
We show in the next section how to further refine this technique, in order to take advantage of rare symbols in each string.

\end{PEDAGOGICALEXAMPLE}
\subsubsection{Refined Analysis:}
\label{sec:InsertDeleteRefinedAnalysis}

We described in the previous section how to take advantage of the fact that some elements appear in one string, but not in the other.  It is natural to wonder if a similar technique can take advantage of cases where a symbol occurs many time in one string, but occurs only once in the other: at some point, the dynamic program will reduce to the case described in the previous section. To be able to notice when this happens, one would need to maintain dynamically the counters of occurrences of each symbol during the execution of the dynamic program, or more simply pre-compute an index on $S$ and $T$ supporting the operators \texttt{rank} and \texttt{select} on it.

Given the support for the \texttt{rank} and \texttt{select} operators on both $S$ and $T$, we can refine the dynamic program to compute the distance $d_{DI}(n,m)$ as follows:
$$
 \left\{
    \begin{array}{ll}
      n \mbox{ if $m==0$;}\\
      m \mbox{ if $n==0$;}\\
      d_{DI}(n-1,m-1)  \mbox{ if $S[n]==T[m]$; } \\
      1+  d_{DI}(n-1,m) \mbox{ if \texttt{rank}$(S,T[m])==0$;} \\
      1+  d_{DI}(n,m-1) \mbox{ if \texttt{rank}$(T,S[n])==0$;} \\
      \min\left\{
      \begin{array}{l}
        1 + d_{DI}(n-1, m-1), \\ % one delete and one insert
        n-\mathtt{select}(S,T[m]) \\\phantom{n}+ d_{DI}(\mathtt{select}(S,T[m],\mathtt{rank}(S,T[m])-1)-1 , m-1), \\ % many deletes
        m-\mathtt{select}(T,S[n]) \\\phantom{m}+ d_{DI}(n-1 , \mathtt{select}(T,S[n],\mathtt{rank}(T,S[n])-1))-1)  % Many insert
      \end{array}
\right\} \mbox{ otherwise.}
    \end{array}
\right.
$$

The running time of the algorithm can then be expressed in function of
the number of recursive calls, 
the number of \texttt{rank} and \texttt{select} operations performed on the strings, 
in order to yield various running times depending upon the solution used to support the \texttt{rank} and \texttt{select} operators.

\begin{theorem} \label{th:DIEditDistanceRankSelect}
Given two strings $S\in[1..\sigma]^n$ and $T\in[1..\sigma]^m$ of respective \texttt{Parikh vectors} 
$(n_a)_{a\in[1..\sigma]}$ and $(m_a)_{a\in[1..\sigma]}$, 
the dynamic program above computes the \textsc{Delete Insert Edit Distance} from $S$ to $T$ 
and the \textsc{Longest Common Sub Sequence} between $S$ and $T$
\begin{enumerate}
\item through at most $4\sum_{a\in[1..\sigma]} n_a m_a$ recursive calls;
\item within $O(\sum_{a\in[1..\sigma]} n_a m_a)$ operations \texttt{rank} or \texttt{select};
\item in time within $O(\sum_{a\in[1..\sigma]} n_a m_a \times \lg(\max_a\{n_a,m_a\}) \times \lg(nm))$ in the comparison model; and 
\item in time within  $O(\sum_{a\in[1..\sigma]} n_a m_a \times \lg\lg\sigma \times \lg(nm))$ in the RAM memory model;
  \end{enumerate}  
\end{theorem}

\begin{proof}
We prove point (1) by an amortization argument. 
\begin{TODO}
Analyze the running time of the new dynamic program
(...)

\end{TODO}
Point (2) is a direct consequence of point (1), given that each recursive call performs a finite number of calls to the \texttt{rank} and \texttt{select} operators.  Point (3) is a simple combination of Point (2) with the classical \emph{inverted posting list} implementation~\cite{1999-BOOK-ManagingGigabytes-WittenMoffatBell} of an index supporting the \texttt{select} operator in constant time and the \texttt{rank} operator via \emph{doubling search}~\cite{1976-IPL-AnAlmostOptimalAlgorithmForUnboundedSearching-BentleyYao}; while point (4) is a simple combination of Point (2) with the index described by Golynski~\etal~\cite{2006-SODA-RankSelectOperationsOnLargeAlphabetsAToolForText-GolynskiMunroRao} to support the \texttt{rank} and \texttt{select} operators. 
\end{proof}

Albeit quite simple, this results corresponds to real improvement in practice: see in Figure~\ref{fig:experimentalResultsDI} how the number of recursive calls is reduced by using such indexes.
Moreover, such a refinement of the analysis and optimization of the computation can be applied to more than the \textsc{Delete Insert Edit Distance}: in the next sections, we describe a similar one for computing the \textsc{Levenshtein Distance} (Section~\ref{sec:DIRUpper}) and the \textsc{Delete Replace} and \textsc{Insert Replace Edit Distance} (Section~\ref{sec:DRUpper}).

\subsection{Levenshtein Distance, or DIR-Edit Distance}
\label{sec:DIRUpper}

In information theory, linguistics and computer science, the \textsc{Levenshtein distance} is a string metric for measuring the difference between two sequences~\cite{2000-SPIRE-ASurveyOfLongestCommonSubsequenceAlgorithms-BergrothHakonenRaita}. It generalizes the \textsc{Delete Insert Edit Distance} explored in the previous section by adding the \texttt{Replace} operator to the operators \texttt{Delete} and \texttt{Insert} (so that it can be also called the \textsc{Delete Insert Replace Edit Distance}, or $DIR$ for short).
The recursion traditionally used is a mere extension from the one described in the previous section\begin{LONG}:
\begin{eqnarray*}
  d_{DIR}(n,m) 
&=& \left\{
    \begin{array}{l}
      m \mbox{ if $n==0$;}\\
      +\infty \mbox{ if $n>m$;}\\
      d_{DIR}(n-1,m-1)  \mbox{ if $S[n]==T[m]$; and } \\
      1 + \min\left\{
      \begin{array}{l}
        d_{DIR}(n,m-1), \\
        d_{DIR}(n-1,m-1) 
      \end{array}
\right\}  \mbox{ otherwise.}
    \end{array}
\right.
\end{eqnarray*}
\end{LONG}\begin{SHORT}, and the adaptive version only a technical extension of the one for the \textsc{Delete Insert Edit Distance}: we will compute the distance $d_{DIR}(n,m)$ as\end{SHORT}
\begin{LONG}
The adaptive version is only a technical extension of the one for the \textsc{Delete Insert Edit Distance}:
% \subsubsection{A Rank and Select Algorithm}
% \label{sec:rankSelectDIR}
% Given the support for the \texttt{rank} and \texttt{select} operators on both $S$ and $T$, we can refine the dynamic program to compute the distance $d_{DIR}(n,m)$ as follows:
\end{LONG}
\begin{SHORT}
$$\left\{
    \begin{array}{ll}
      n \mbox{ if $m==0$;}\\
      m \mbox{ if $n==0$;}\\
      d_{DIR}(n-1,m-1)  \mbox{ if $S[n]==T[m]$; } \\
      1+  d_{DIR}(n-1,m-1) \mbox{ 
      \begin{tabular}{l}
if \texttt{rank}$(S,T[m])==0$  \\
and \texttt{rank}$(T,S[n])==0$ ;        
      \end{tabular}
} \\
      1+  d_{DIR}(n-1,m) \mbox{ 
      \begin{tabular}{l}
if \texttt{rank}$(S,T[m])==0$ \\
but \texttt{rank}$(T,S[n])>0$ ;        
      \end{tabular}
} \\
      1+  d_{DIR}(n,m-1) \mbox{ 
      \begin{tabular}{l}
if \texttt{rank}$(T,S[n])==0$ \\
but \texttt{rank}$(S,T[m])>0$ ;
      \end{tabular}} \\
      \min\left\{
      \begin{array}{l}
        n-\mathtt{select}(S,T[m]) \\\phantom{n}+ d_{DIR}(\mathtt{select}(S,T[m],\mathtt{rank}(S,T[m])-1)-1 , m-1), \\ % Keeping T[m], many deletes in S
        m-\mathtt{select}(T,S[n]) \\\phantom{m}+ d_{DIR}(n-1 , \mathtt{select}(T,S[n],\mathtt{rank}(T,S[n])-1))-1),  \\% Keeping S[n] Many insert
        1+ d_{DIR}(n-1,m-1)  
      \end{array}
\right\} \mbox{ otherwise.}
    \end{array}
\right.
$$
\end{SHORT}

\begin{LONG}
$$\left\{
    \begin{array}{ll}
      n \mbox{ if $m==0$;}\\
      m \mbox{ if $n==0$;}\\
      d_{DIR}(n-1,m-1)  \mbox{ if $S[n]==T[m]$; } \\
      1+  d_{DIR}(n-1,m-1) \mbox{ 
      \begin{tabular}{l}
if \texttt{rank}$(S,T[m])==0$  \\
and \texttt{rank}$(T,S[n])==0$ \mbox{ (REPLACE)};        
      \end{tabular}
} \\
      1+  d_{DIR}(n-1,m) \mbox{ 
      \begin{tabular}{l}
if \texttt{rank}$(S,T[m])==0$ \\
but \texttt{rank}$(T,S[n])>0$ \mbox{ (DELETE)};        
      \end{tabular}
} \\
      1+  d_{DIR}(n,m-1) \mbox{ 
      \begin{tabular}{l}
if \texttt{rank}$(T,S[n])==0$ \\
but \texttt{rank}$(S,T[m])>0$ (INSERT);
      \end{tabular}} \\
      \min\left\{
      \begin{array}{l}
        n-\mathtt{select}(S,T[m]) \\+ d_{DIR}(\mathtt{select}(S,T[m],\mathtt{rank}(S,T[m])-1)-1 , m-1) \mbox{ (DELETE) }, \\ % Keeping T[m], many deletes in S
        m-\mathtt{select}(T,S[n]) \\+ d_{DIR}(n-1 , \mathtt{select}(T,S[n],\mathtt{rank}(T,S[n])-1))-1) \mbox{ (INSERT) },  \\% Keeping S[n] Many insert
        1+ d_{DIR}(n-1,m-1)  \mbox{ (REPLACE) }
      \end{array}
\right\} \mbox{ otherwise.}
    \end{array}
\right.
$$
\end{LONG}

% \subsubsection{Refined Analysis}
% \label{sec:DeleteDeleteReplaceRefinedAnalysis}

The refined analysis yields similar results (we omit the proof for lack of space):

\begin{theorem} \label{th:DIREditDistanceRankSelect}
Given two strings $S\in[1..\sigma]^n$ and $T\in[1..\sigma]^m$ of respective \texttt{Parikh vectors} 
$(n_a)_{a\in[1..\sigma]}$ and $(m_a)_{a\in[1..\sigma]}$, 
the dynamic program above computes the \textsc{Levenshtein Edit Distance} from $S$ to $T$ 
\begin{enumerate}
\item through at most $4\sum_{a\in[1..\sigma]} n_a m_a$ recursive calls;
\item within $O(\sum_{a\in[1..\sigma]} n_a m_a)$ operations \texttt{rank} or \texttt{select};
\item in time within $O(\sum_{a\in[1..\sigma]} n_a m_a \times \lg(\max_a\{n_a,m_a\}) \times \lg(nm))$ in the comparison model; and 
\item in time within  $O(\sum_{a\in[1..\sigma]} n_a m_a \times \lg\lg\sigma \times \lg(nm))$ in the RAM memory model;
  \end{enumerate}  
\end{theorem}

It is important to note that for two strings $S$ and $T$, the computation of the \textsc{Levenshtein Edit Distance} from $S$ to $T$ actually generates more recursive calls than the computation of the \textsc{Delete Insert Edit Distance} from $S$ to $T$, but that the analysis above does not capture this difference.
In the following section, we project this result to two equivalent edit distances, the \textsc{Delete Replace} and \textsc{Insert Replace Edit Distances}, for which the dynamic program explores only half of the position in the dynamic program matrix compared to the \textsc{Levenshtein Edit Distance} or \textsc{Delete Insert Edit Distance}.

\subsection{DR-Edit Distance and IR-Edit Distance}
\label{sec:DRUpper}

Given a source string $S\in[1..\sigma]^n$ and a target string $T\in[1..\sigma]^m$, the \textsc{Delete Replace Edit Distance} from $S$ to $T$ and the \textsc{Insert-Replace Edit Distance} from $T$ to $S$ are the same, as the sequence of \texttt{Insert} or \texttt{Replace} operations transforming $S$ into $T$ is the symmetric to the sequence of \texttt{Delete} or \texttt{Replace} operations transforming $T$ back into $S$. 
%
% We describe  the classical dynamic program which compute the \textsc{Delete Replace Edit Distance} from $S$ to $T$, along with some trivial degenerate cases, which will yield a deeper analysis of the computational problem

% \subsubsection{Classical solution:}
% \label{sec:DeleteReplaceClassical}
% Given two strings $S\in[1..\sigma]^n$ and $T\in[1..\sigma]^m$, let's note $d_{DR}(n,m)$ the \textsc{Delete Replace Edit Distance} from $S$ to $T$. 

As before, if the last symbols of $S$ and $T$ match, the edit distance is the same as the edit distance between the prefixes of respective lengths $n-1$ and $m-1$ of $S$ and $T$. Otherwise, the edit distance is the minimum between the edit distance when inserting a copy of the last symbol of $T$ in $S$ (i.e. deleting this symbol in $T$) and the edit distance when replacing the mismatching symbol in $S$ by the corresponding one in $T$. 
\begin{LONG}
More formally:
\begin{eqnarray*}
  d_{DR}(S[1..n],T[1..m]) 
&=& \left\{
    \begin{array}{l}
      m \mbox{ if $n==0$;}\\
      +\infty \mbox{ if $n>m$;}\\
      d_{DR}(S[1..n-1],T[1..m-1])  \mbox{ if $S[n]==T[m]$; and } \\
      1 + \min\left\{
      \begin{array}{l}
        d_{DR}(S[1..n],T[1..m-1]), \\
        d_{DR}(S[1..n-1],T[1..m-1]) 
      \end{array}
      \right\}  \mbox{ otherwise.}
    \end{array}
  \right.
\end{eqnarray*}

\end{LONG}

\begin{LONG}
Using a few more optimizations than in Section~\ref{sec:DeleteInsertClassical}, this recursive definition yields an algorithm to compute the \textsc{Insert Replace Edit Distance} from $S$ to $T$ in time within $O(m^2)$ and space within $O(m)$.
One can observe a few optimizations, such as that the edit distance can be computed in time linear in $m$ as soon as $n$ is equal to $m$, as no further \texttt{Insert} operation can be performed.
\end{LONG}
%
% We describe in the next section a technique taking advantage of the discrepancies between the \texttt{Parikh vectors} of $S$ and $T$.

% \subsubsection{Refined Analysis}
% \label{sec:DeleteReplaceRefinedAnalysis}

As in the two previous sections, given the support for the \texttt{rank} and \texttt{select} operators on both $S$ and $T$, we can refine the dynamic program to compute the \textsc{Delete Replace Edit Distance} $d_{DR}(n,m)$ as follows:
\begin{LONG}
$$
\left\{
    \begin{array}{ll}
      n \mbox{ if $m==0$;}\\
      \infty \mbox{ if $n<m$;}\\
      d_{DR}(n-1,m-1)  \mbox{ if $S[n]==T[m]$; } \\
      1+  d_{DR}(n-1,m) \mbox{ if \texttt{rank}$(T,S[n])==0$ (DELETE);} \\ %DELETE
      1+  d_{DR}(n-1,m-1) \mbox{ if \texttt{rank}$(S,T[n])==0$ (REPLACE);} \\ %REPLACE
      \min\left\{
      \begin{array}{l}
        n-\mathtt{select}(S,T[m]) \\+ d_{DR}(\mathtt{select}(S,T[m],\mathtt{rank}(S,T[m])-1)-1 , m-1) \mbox{ (DELETE), } \\ % Keeping T[m], many deletes in S
        1 + d_{DR}(n-1,m-1)  \mbox{ (REPLACE) } % REPLACING S[n] by a copy of T[m]
      \end{array}
\right\} \mbox{ otherwise.}
    \end{array}
\right.
$$
\end{LONG}
\begin{SHORT}
$$
\left\{
    \begin{array}{ll}
      n \mbox{ if $m==0$;}\\
      \infty \mbox{ if $n<m$;}\\
      d_{DR}(n-1,m-1)  \mbox{ if $S[n]==T[m]$; } \\
      1+  d_{DR}(n-1,m) \mbox{ if \texttt{rank}$(T,S[n])==0$ ;} \\ %DELETE
      1+  d_{DR}(n-1,m-1) \mbox{ if \texttt{rank}$(S,T[n])==0$ ;} \\ %REPLACE
      \min\left\{
      \begin{array}{l}
        n-\mathtt{select}(S,T[m]) \\ \phantom{n}+ d_{DR}(\mathtt{select}(S,T[m],\mathtt{rank}(S,T[m])-1)-1 , m-1), \\ % Keeping T[m], many deletes in S
        1 + d_{DR}(n-1,m-1)   % REPLACING S[n] by a copy of T[m]
      \end{array}
\right\} \mbox{ otherwise.}
    \end{array}
\right.
$$
\end{SHORT}
 
The analysis from the two previous sections projects to a similar result.

\begin{theorem} \label{th:DREditDistanceRankSelect}
Given two strings $S\in[1..\sigma]^n$ and $T\in[1..\sigma]^m$ of respective \texttt{Parikh vectors} 
$(n_a)_{a\in[1..\sigma]}$ and $(m_a)_{a\in[1..\sigma]}$, 
the dynamic program above computes the \textsc{Delete Replace Edit Distance} from $S$ to $T$ 
\begin{enumerate}
\item through at most $4\sum_{a\in[1..\sigma]} n_a m_a$ recursive calls;
\item within $O(\sum_{a\in[1..\sigma]} n_a m_a)$ operations \texttt{rank} or \texttt{select};
\item in time within $O(\sum_{a\in[1..\sigma]} n_a m_a \times \lg(\max_a\{n_a,m_a\}) \times \lg(nm))$ in the comparison model; and 
\item in time within  $O(\sum_{a\in[1..\sigma]} n_a m_a \times \lg\lg\sigma \times \lg(nm))$ in the RAM memory model;
  \end{enumerate}  
\end{theorem}

Parameterizing the analysis of the computation of the \textsc{Longest Common Sub Sequence}, of the \textsc{Levenshtein Edit Distance} and of the \textsc{Delete Replace} or \textsc{Insert Replace Edit Distance} would be only of moderate theoretical interest, if it did not correspond to some correspondingly ``easy'' instances in practice. In the next section we describe some preliminary experimental results which seem to indicate the existence of such ``easy'' instances in information retrieval.

\section{Experiments}
\label{sec:experiments}

In order to test the practicality of the parameterization and algorithms described in the previous section, we performed some preliminary experiments on some public data sets from the \textsc{Gutemberg} project~\cite{gutembergProject}. 
\begin{LONG}
We describe the data set and experimental setup in Section~\ref{sec:datasets}, and the preliminary results and their interpretation in Section~\ref{sec:experimentalResults}.
\end{LONG}
\begin{LONG}
\subsection{Data Sets}
\label{sec:datasets}
\end{LONG}
\begin{LONG}
Started by Michael Hart in 1971~\cite{gutembergProjectWikipedia}, the \textsc{Gutemberg} project gathers electronic copies of public domain books, and as such is a publically available data set for testing algorithms on real text.
\end{LONG}
We considered each text as a sequence of words (hence considering as equivalent all the word separations, from blank spaces to punctuations and line jumps), which results in large alphabets.
Due to some problems with the implementation, we could not run the algorithms for texts larger than 32kB (a memory issue with a library in \texttt{Python}), so we extracted the first 32kB of the texts
 ``Romeo \& Juliet'' (English),
 ``Romeo \& Julia'' (German),
 ``Hamlet'' (German), and 
 ``Punch or the London chivalry vol 99'' (English);
 the last text being a randomly picked non Shakespeare text.

 \begin{LONG}
 \subsection{Experimental Results}
 \label{sec:experimentalResults}
 \end{LONG}

Figures~\ref{fig:experimentalResultsDI}, \ref{fig:experimentalResultsDIR} and~\ref{fig:experimentalResultsDR} show the number of recursive calls from the main recursive function for four pairs of texts for each algorithm described in Section~\ref{sec:algorithms}: 
 ``Romeo \& Juliet'' (English) vs  ``Punch or the London chivalry vol 99'' (English),
 ``Romeo \& Juliet'' (English) vs  ``Romeo \& Julia'' (German),  
``Romeo \& Juliet'' (English) vs ``Hamlet'' (English), and 
``Romeo \& Julia'' (German) vs ``Hamlet'' (German).

\begin{figure}
\begin{minipage}[t]{.9\linewidth}
\includegraphics[width=.9\linewidth]{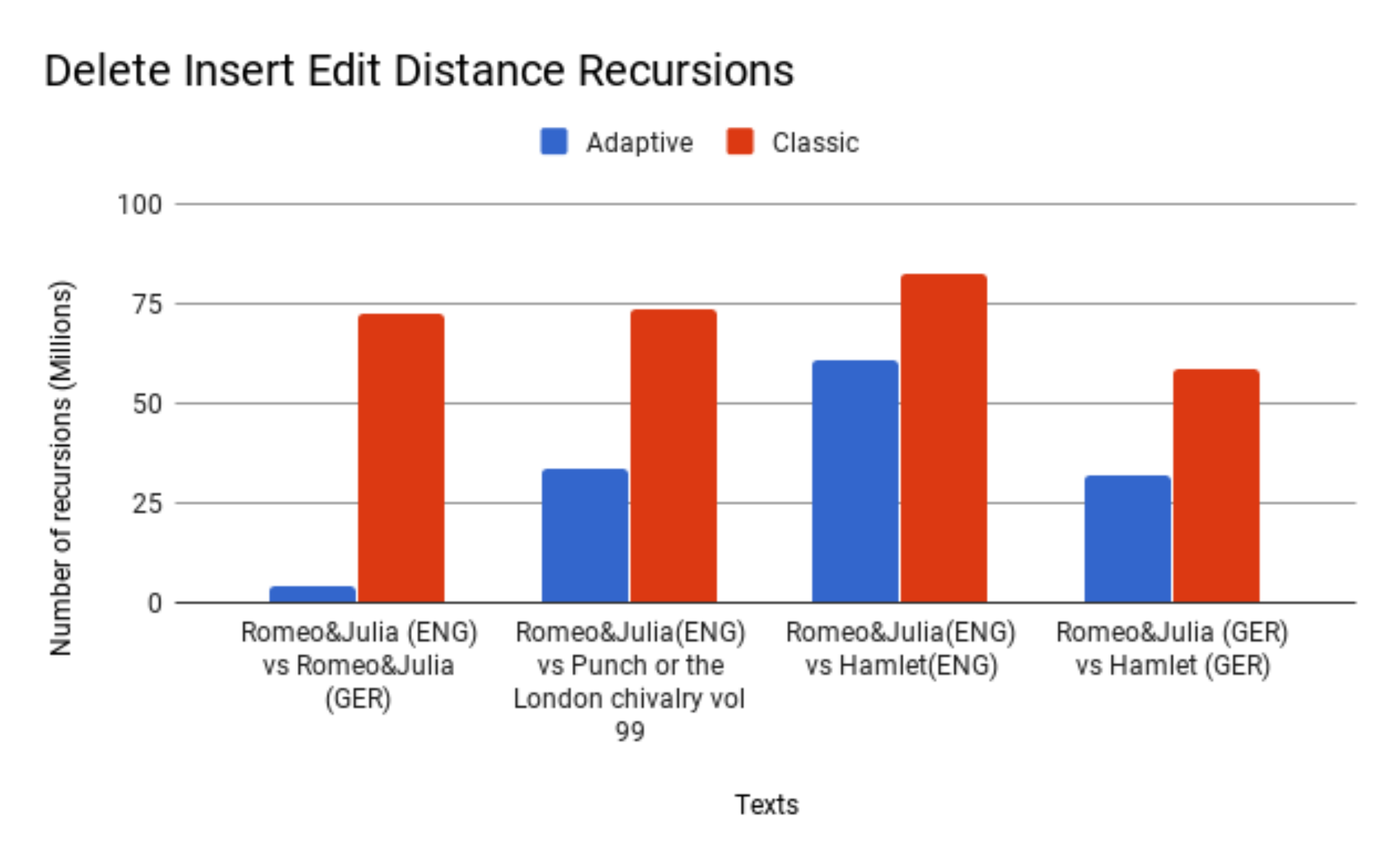}
\begin{LONG}
\caption{Experimental results for the computation of the \textsc{Delete Insert Edit Distance} (and, by extention, of the computation of the \textsc{Longest Common Sub Sequence}) by the adaptive algorithm described in Section \ref{sec:DIUpper}.\label{fig:experimentalResultsDI}}
\end{LONG}
\begin{SHORT}
\caption{Experimental results for the \textsc{Delete Insert Edit Distance}.\label{fig:experimentalResultsDI}}
\end{SHORT}
\end{minipage}
\hfill
\begin{minipage}[t]{.9\linewidth}
\includegraphics[width=.9\linewidth]{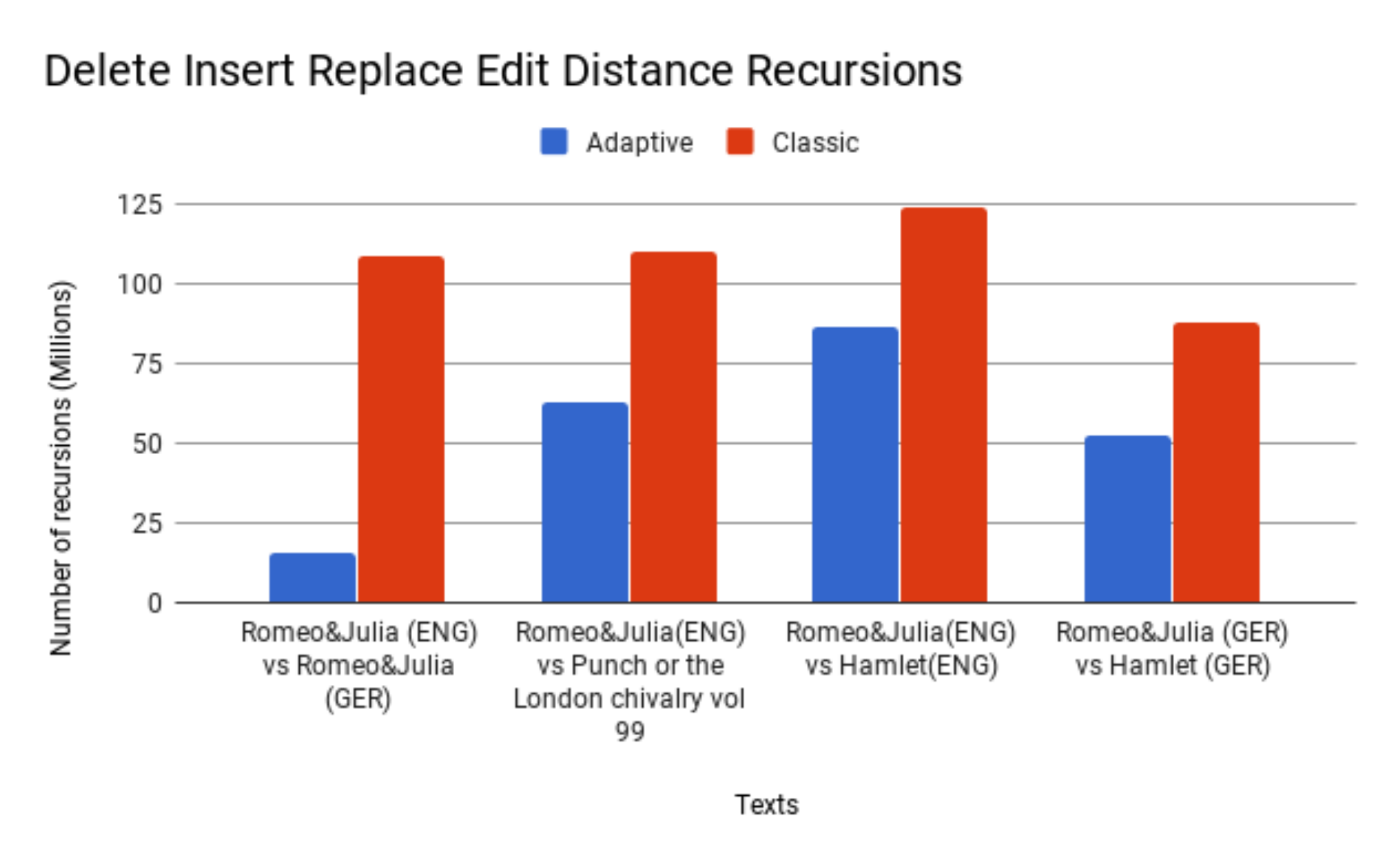}
\begin{LONG}
\caption{Experimental results for the computation of the \textsc{Levenshtein Edit Distance} by the adaptive algorithm described in Section~\ref{sec:DIRUpper}.\label{fig:experimentalResultsDIR}}
\end{LONG}
\begin{SHORT}
\caption{Experimental results for the \textsc{Levenshtein Edit Distance}.\label{fig:experimentalResultsDIR}}
\end{SHORT}
\end{minipage}
\hfill
\begin{minipage}[t]{.9\linewidth}
\includegraphics[width=.9\linewidth]{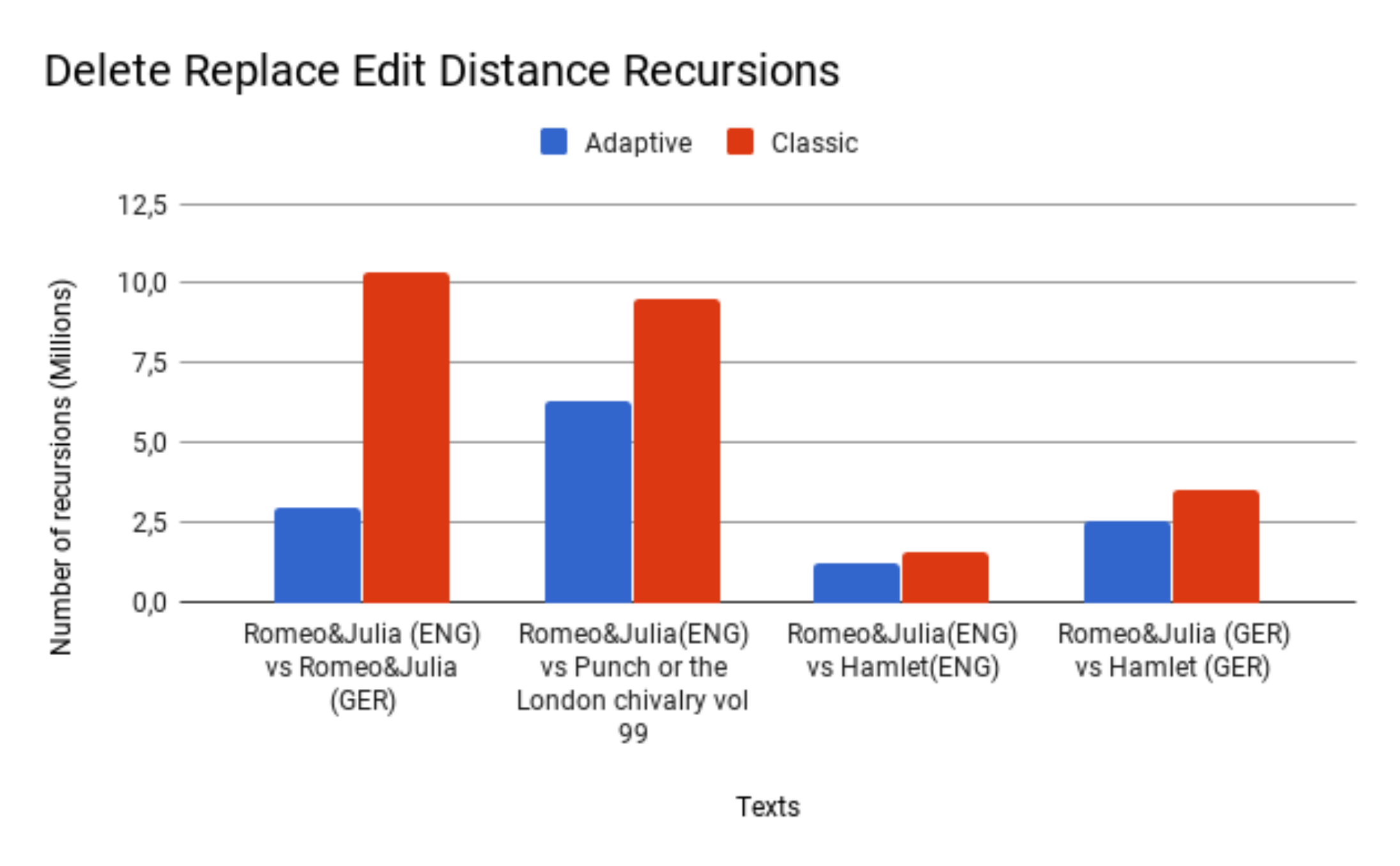}
\begin{LONG}
\caption{Experimental results for the computation of the \textsc{Delete Replace Edit Distance} (and, by extension, of the computation of the \textsc{Insert Replace Edit Distance}) by the adaptive algorithm described in Section~\ref{sec:DRUpper}.\label{fig:experimentalResultsDR}}
\end{LONG}
\begin{SHORT}
\caption{Experimental results for the \textsc{Delete Replace Edit Distance}.\label{fig:experimentalResultsDR}}
\end{SHORT}
\end{minipage}
\end{figure}

For the three types of \textsc{Edit Distances} and the four pairs of texts, the adaptive variants perform less recursive calls.  
For the three types of \textsc{Edit Distances}, the difference in the number of recursive calls is less between the two texts  from the same author (i.e. ``Romeo \& Juliet'' (English) vs ``Hamlet'' (English) and ``Romeo \& Julia'' (German) vs ``Hamlet'' (German)), because the vocabulary is the same, and is the most between texts of distinct languages (i.e. ``Romeo \& Juliet'' (English) vs ``Romeo \& Julia'' (German)), because the vocabulary (i.e. the alphabet) is mostly distinct. Still, for two texts in the same language, but from distinct authors (i.e.  ``Romeo \& Juliet'' (English) vs ``Punch or the London chivalry vol 99'' (English)), the difference is quite sensible.

Obviously, those experimental results are only preliminary, and a more thorough study is needed (and underway), both with a larger data set and with a larger range of measures, from the \emph{running time} with various indexing data structures supporting the operators \texttt{rank} and \texttt{select}, to the number of entries of the dynamic program matrix being effectively computed.
We discuss additional perspectives for future work in the next section.

\begin{CONDITIONAL}
\section{Conditional Lower Bounds}
\label{sec:lowerBounds}

\subsection{LCSS or DI-Edit Distance}
\label{sec:DILower}

\subsection{DIR-Edit Distance}
\label{sec:DIRLower}

\subsection{DR-Edit Distance}
\label{sec:DRLower}

\end{CONDITIONAL}

\section{Discussion}
\label{sec:discussion}

We have shown how the computation of other \textsc{Edit Distances} than the \textsc{Insert Swap} and \textsc{Delete Swap Edit Distance} is also sensitive to the \texttt{Parikh vectors} of the input. We discuss here various directions in which these results can be extended, from the possibility of proving conditional lower bounds in the refined analysis model, to further refinements of the analysis for these same \textsc{Edit Distances}, and to the analysis of other dynamic programs.

\subsubsection{Adaptive Conditional Lower Bounds:}
\label{sec:ConditionalLowerBounds}

Backurs and Indyk~\cite{2015-STOC-EditDistanceCannotBeComputedInStronglySubquadraticTimeUnlessSETHIsFalse-BackursIndyk} showed that the $O(n^2)$ upper bound for the computation of the \textsc{Delete Insert Replace Edit Distance} is optimal unless the \emph{Strong Exponential Time Hypothesis} (SETH) is false, and since then the technique has been applied to various other related problems. 
Should the reduction from the SETH to the \textsc{Edit Distance} computation be refined as shown here for the upper bound, it would speak in favor of the optimality of the analysis.

\subsubsection{Other measures of difficulty:}
\label{sec:other-meas-diff}

Abu-Khzam et al.~\cite{2011-DO-ChargeAndReduceAFixedParameterAlgorithmForStringToStringCorrection-AbukhzamFernauLangstonLeeculturaStege} described an algorithm computing the \textsc{Insert Swap Edit Distance}  $d$ from $S$ to $T$ in time within $O({1.6181}^d m)$, which is polynomial in the size of the input if exponential in the output size $d$.  The output distance $d$ itself can be as large as $n$, but such instances are not necessarily difficult: Barbay and P\'erez-Lantero~\cite{2015-SPIRE-AdaptiveComputationOfTheSwapInsertCorrectionDistance-BarbayPerez} showed that the gap vector between the \texttt{Parikh vectors} separate the hard instances from the easy one, and we showed that the same applied to other \textsc{Edit Distances}.

But still, among instances of fixed input size, output distance, and imbalance between the \texttt{Parikh vectors}, there are instances easier than others (e.g. the computation of the \textsc{Insert Swap Edit Distance} on an instance where all the \texttt{insertions} are in the left part of $S$ while all the \texttt{swaps} are in the right part of $S$). A measure which to refine the analysis would be the cost of encoding a \emph{certificate} of the \textsc{Edit Distance}, one which is easier to check than recomputing the distance itself.

\subsubsection{Indexed Dynamic Programming:}
\label{sec:similarProblems}

Our results are close in spirit to those in fixed-parameter complexity, but with an important difference, namely, trying to spot one or more parameters that explain what makes an instance hard or easy.
For the computation of the \textsc{Insert Swap} and \textsc{Delete Swap Edit Distances}, the size of the alphabet $d$ makes the difference between polynomial time and NP-hardness. However, Barbay and P\'erez-Lantero~\cite{2015-SPIRE-AdaptiveComputationOfTheSwapInsertCorrectionDistance-BarbayPerez} showed that different instances of the same size can exhibit radically different costs-for a given fixed algorithm. The parameterized analysis captures in parameters the cause for such cost differences.  We described how the same logic applies to other types of \textsc{Edit Distances}, and it is likely that similar situations happen with many other algorithms based on dynamic programming, such as the computation of the \textsc{Fr\'echet Distance}~\cite{1995-IJCGA-ComputingTheFrechetDistanceBetweenTwoPolygonalCurves-AltGodau}, the \textsc{Discrete Fr\'echet Distance}~\cite{1994-TR-ComputingDiscreteFrechetDistance-EiterMannila} and the decision problem \textsc{Orthogonal Vector}~\cite{2014-FOCS-WhyWalkingTheDogTakesTimeFrechetDistanceHasNoStronglySubquadraticAlgorithmsUnlessSETHFails-Bringmann}.

\medskip \textbf{Acknowledgments:} The author would like to thank Pablo P\'erez-Lantero for introducing the problem of computing the \textsc{Edit Distance} between strings; Felipe Lizama for a semester of very interesting discussions about this approach; and an anonymous referee from the journal Transaction on Algorithms for his positive feedback and encouragement.
\textbf{Funding:} J\'er\'emy Barbay is partially funded by the project Fondecyt Regular no. 1170366 from Conicyt.
%
% \textbf{Competing Interests:} The authors declare that they have no competing financial interests relevant to the material exposed in this article.
%
\textbf{Data and Material Availability:} The source of this article, along with the code and data used for the experiments described within, will be made publicly available upon publication at the url \url{https://github.com/FineGrainedAnalysis/EditDistances}.
%
% ---- Bibliography ----

\bibliographystyle{splncs04} 
\bibliography{/home/jbarbay/EverGoing/Bibliography/bibliographyDatabaseJeremyBarbay,/home/jbarbay/EverGoing/Publications/publications-ExportedFromOrgmode-Barbay}

\newpage
\appendix
\begin{center}
{\huge \textbf{APPENDIX}}
\end{center}

In this appendix, we briefly discuss some minor topics, such as how the algorithm \texttt{Local Insertion Sort} described and analyzed by Moffat and Petersson~\cite{1992-ACJ-AnOverviewOfAdaptiveSorting-MoffatPetersson} combined with an index supporting the \texttt{rank} and \texttt{select} operators potentially yields a faster computation of the \textsc{Swap Edit Distance} (Section~\ref{sec:swap}), or how to combine techniques that take advantage of the \texttt{Parikh vectors} of the input strings with techniques that take advantage of the output distance (Section~\ref{sec:distanceAdaptivity}).

\section{Adaptive Computation of the \textsc{Swap Edit Distance}}
\label{sec:swap}

Out of the $2^4-1=15$ non trivial distances which can be obtained from the four operators \texttt{Delete}, \texttt{Insert}, \texttt{Replace} and \texttt{Swap}, the \textsc{Swap Edit Distance} is the simplest for which no linear time algorithm is known: Wagner and Lowrance~\cite{1975-JACM-AnExtensionOfTheStringToStringCorrectionProblem-WagnerLowrance} described a dynamic program to compute it in time within $O(n^2)$. As the \texttt{swap} operator does not remove nor insert any symbol, the \textsc{Swap Edit Distance} between two strings of distinct lengths, or of same lengths but with distinct \texttt{Parikh vectors}, is always infinite. Such cases can be checked in time within $O(n+m+\sigma)$, and can be ignored as degenerated cases.

The computation of the \textsc{Swap Edit Distance} between two strings $S,T\in[1..\sigma]^n$ of same \texttt{Parikh vector} $(n_i)_{i\in[1..\sigma]}$ is quite similar to the problem of \textsc{Sorting} one string into the other via the exchange of consecutive elements: the \texttt{swap} operator is merely reordering the symbols of the strings, and the combined actions of all the \texttt{swap} operations can be summarized by a simple permutation over $[1..n]$. Consider the shortest sequence of such \texttt{swap} operators transforming $S$ into $T$, and $\pi$ the corresponding permutation. The number $d$ of inversions in $\pi$, defined as the number of pairs $i,j\in[1..n]$ of positions $i<j$ such that $\pi[j]<\pi[j]$ (i.e. the order of $(i,j)$ is inverse to that of $(\pi[i],\pi[j])$), is exactly the number $d$ of \texttt{swaps} required to ``reorder'' $S$ into $T$, i.e. the \textsc{Swap Edit Distance} $d$ from $S$ to $T$.

Moffat and Petersson~\cite{1992-ACJ-AnOverviewOfAdaptiveSorting-MoffatPetersson} described two sorting algorithms adaptive to the number $d$ of inversions of the input.  The first one is the classical sorting algorithm \texttt{Insertion Sort}, which sorts an array $A$ with $d$ inversions using only the operator \texttt{swaps}, using $O(d)\subseteq O(n^2)$ comparisons and \texttt{swaps}. The second one is the adaptive algorithm \texttt{Local Insertion Sort}, based on a \texttt{Finger Search Tree}, which uses only $O(n(1+\lg(d/n)))$ comparisons and can be easily modified to \emph{count the number $d$ of inversions}, i.e. the \textsc{Swap Edit Distance} $d$ between an array $A$ and its sorted version.  We describe below how to take advantage of this algorithm to compute the \textsc{Swap Edit Distance} between two arbitrary strings.

First, consider the one-to-one mapping between positions in $S$ and positions in $T$:
\begin{lemma}
Given two strings $S,T\in[1..\sigma]^n$ of same \texttt{Parikh vector}, the $i$-th symbol $\alpha$ in $S$ is mapped to the $i$-th symbol $\alpha$ in $T$ by the shortest sequence of \texttt{Swap operations transforming $S$ into $T$.}
\end{lemma}
\begin{proof}
As \texttt{Swap} is the only operator available, no symbol is added or removed from $S$ to obtain $T$, so that there is a one to one mapping between each symbol of $S$ and each symbol of $T$. Moreover, any sequence of \texttt{Swap} operation matching the $i$-th symbol $\alpha$ in $T$ to the $j$-th occurrence of $\alpha$ in $T$ can be made shorter if $j\neq i$.
\end{proof}

Then, consider how to compute the \textsc{Swap Edit Distance} using such a mapping and an index supporting the \texttt{rank} and \texttt{select} operators:
\begin{theorem} \label{th:SwapRankSelect} Given two strings $S,T\in[1..\sigma]^n$ of same \texttt{Parikh vector} $(n_i)_{i\in[1..\sigma]}$, there is an algorithm computing the \textsc{Swap Edit Distance} from $S$ to $T$ via $O(n(1+\lg(d/n))$ \texttt{rank} and \texttt{select} operations.
\end{theorem}
\begin{proof}
Define the following process to decide if the symbols at positions $i$ and $j$ in $S$ must be inversed during the transformation of $S$ into $T$ minimizing the number of \texttt{Swap} operations:
Let $a=S[i]$ and $b=S[j]$, so that $i$ and $j$ are the $\mathtt{rank}(S,a,i)$-th and $\mathtt{rank}(S,b,j)$-th occurences of $a$ and $b$ in $S$, respectively.
Let $i'=\mathtt{select}(T,a,\mathtt{rank}(S,a,i))$ and $j'=\mathtt{select}(T,b,\mathtt{rank}(S,b,j))$ be the positions of the corresponding occurences in $T$.
The symbols will be inversed during the transformation of $S$ into $T$ if and only if $(i,j)$ and $(i',j')$'s orders are inversed.

``Sorting'' $S$ into $T$ using the algorithm \texttt{Local Insertion Sort} described by Moffat and Petersson~\cite{1992-ACJ-AnOverviewOfAdaptiveSorting-MoffatPetersson} and the process described above to answer comparisons between $\pi(i)$ and $\pi(j)$ yields an algorithm computing the \textsc{Swap Edit Distance} using within $O(n(1+\lg(d/n))$ \texttt{rank} and \texttt{select} operarions. 
\end{proof}

This yields as many solutions as there are data structures to support the \texttt{rank} and \texttt{select} operators, each yielding a distinct computational tradeoff on the previous lemma: we describe two.  
The first one is based on \emph{inverted posting lists}~\cite{1999-BOOK-ManagingGigabytes-WittenMoffatBell} and an amortized analysis of \texttt{doubling search} algorithm~\cite{1976-IPL-AnAlmostOptimalAlgorithmForUnboundedSearching-BentleyYao}:

\begin{corollary}
Given two strings $S,T\in[1..\sigma]^n$ of same \texttt{Parikh vector} $(n_i)_{i\in[1..\sigma]}$, there is an algorithm computing the \textsc{Swap Edit Distance} from $S$ to $T$ in time within $O(n(1+\lg(d/n))\lg(n))\subset O(n^2)$ in the comparison based decision tree model.
\end{corollary}
\begin{proof}
A simple combination of Theorem~\ref{th:SwapRankSelect} with the classical \emph{inverted posting list} implementation~\cite{1999-BOOK-ManagingGigabytes-WittenMoffatBell} of an index supporting the \texttt{select} operator in constant time and the \texttt{rank} operator via \emph{doubling search}\cite{1976-IPL-AnAlmostOptimalAlgorithmForUnboundedSearching-BentleyYao}.
\end{proof}

Note that it should be possible to refine the analysis, as $O(\lg n)$ is a very crude upper bound on the complexity of supporting \texttt{rank} or \texttt{select}, one should be able 
to express it in function of $n_\alpha$, and
to amortize it over all \texttt{rank} and \texttt{select} operations for each symbol $\alpha$.

The second one is based on the more sophisticated succinct data structure described by Golynski~\etal~\cite{2006-SODA-RankSelectOperationsOnLargeAlphabetsAToolForText-GolynskiMunroRao}:
\begin{corollary}
Given two strings $S,T\in[1..\sigma]^n$, there is an algorithm computing the \textsc{Swap Edit Distance} from $S$ to $T$ in time within $O(n(1+\lg(d/n))\lg\lg\sigma)$.
\end{corollary}
\begin{proof}
A simple combination of Theorem~\ref{th:SwapRankSelect} with the index described by Golynski~\etal~\cite{2006-SODA-RankSelectOperationsOnLargeAlphabetsAToolForText-GolynskiMunroRao} to support the \texttt{rank} and \texttt{select} operators.
\end{proof}

Next, we discuss the minor topic of combining techniques that take advantage of the \texttt{Parikh vectors} of the input strings with techniques that take advantage of the output distance.

\section{Distance Adaptive Computation for all Edit Distances}
\label{sec:distanceAdaptivity}

Abu-Khzam et al.~\cite{2011-DO-ChargeAndReduceAFixedParameterAlgorithmForStringToStringCorrection-AbukhzamFernauLangstonLeeculturaStege} described an algorithm computing the \textsc{Insert Swap Edit Distance} $d$ from $S$ to $T$ in time within $O({1.6181}^d m)$, which is adaptive to the distance $d$.  We show here that a similar technique can be applied to the other edit distances based on the operators \texttt{Delete}, \texttt{Insert} and \texttt{Replace}, so that to obtain a complexity adaptive to the distance $d$ being computed (Section~\ref{sec:distanceAdaptiveComputation}) and to combine this technique with the ones described previously (Section~\ref{sec:CombinationWithOtherAdaptiveTechniques}).

\subsection{Distance Adaptive Computation}
\label{sec:distanceAdaptiveComputation}

\begin{lemma}
For any edit distance based on a subset of the set of operators $\{$\texttt{Delete}, \texttt{Insert}, \texttt{Replace} $\}$, there is an algorithm which checks that the distance $d$ from a source string $S\in[1..\sigma]^n$ to a target string $T\in[1..\sigma]^m$ is smaller than a promise $D\in[0..\max\{n,m\}]$ (i.e. if $d\leq D$) in time within $O(D\min\{n,m\})$ and space within $O(n+m)$.
\end{lemma}

\begin{theorem}
For any edit distance  based on a subset of the set of operators $\{$\texttt{Delete}, \texttt{Insert}, \texttt{Replace} $\}$, 
there is an algorithm which computes this distance $d$ from a source string $S\in[1..\sigma]^n$ to a target string $T\in[1..\sigma]^m$ 
in time within $O(d\min\{n,m\})$ and space within $O(n+m)$.
\end{theorem}

\begin{INUTILE}
\begin{figure}
\centering
\begin{minipage}{.3\linewidth}
$$
\begin{array}{c|*{9}{c}}
      & 1        &          &    &                          & i        &          &          & n        \\ \hline
  1   & d(1,1)                                                                                                       \\
               && d(2,2) &                                                                                              \\
                          &&& \dots                                                                                    \\
i     &          &          &        &               & d(i,i)        &          &          &          \\
 &&&&&& \dots                                                                                                              \\
                                                                                                              \\
 m    &                                                                                                       \\
\end{array}$$
\end{minipage}
\caption{If the edit distance from $S$ to $T$ is at most $D$, it is not necessary to compute ``cells'' at distance more than $D$ from the diagonal. \label{fig:diagonal}}
\end{figure}
\end{INUTILE}

\subsection{Combination with Other Adaptive Techniques}
\label{sec:CombinationWithOtherAdaptiveTechniques}

\begin{corollary}
There is an algorithm which computes this \textsc{Delete Insert Edit Distance} $d$ from a source string $S\in[1..\sigma]^n$ to a target string $T\in[1..\sigma]^m$ in time within $O(d\min\{n',m'\})$ and space within $O(n+m)$, where $n'=\sum_{\alpha, m_\alpha>0}n_\alpha$ and $m'=\sum_{\alpha, n_\alpha>0}m_\alpha$.
\end{corollary}

\begin{TODO}
ADD corollaries for the other distances.
\end{TODO}

\begin{PYTHON}
  \section{Python implementation of Adaptive Algorithms}
  \label{sec:pythonImplementation}
  
  \subsection{LCSS or DI-Edit Distance}
  \label{sec:DIExperiment}
  
  \begin{algorithm}[p]
  \caption{Classical algorithm to compute the \textsc{Delete-Insert Edit Distance} between two strings. The implementation is decomposed in two parts:  the \texttt{computation} function initializes the array of values, which is filled recursively by the \texttt{recursive} function. For the sake of space, the documentation strings and unit tests were not included, but the source file including those is available at \url{https://github.com/FineGrainedAnalysis/EditDistance}.
    \label{algo:classicalDI}} 
  \begin{python} 
def classic_di(base, i, j, source, target):
    val = 0
    # Base cases
    if j < 0:
        return i+1
    elif i < 0:
        return j+1
    # Check if the cell has been filled
    # default cell value is -1 
    if base.get(i, j) < 0:
        # Get ith and jth word froum source and target
        if source.get(i) == target.get(j):
            val = classic_di(base, i-1, j-1, source, target)
        else:
            delete_dist = classic_di(base, i-1, j, source, target)
            insert_dist = classic_di(base, i, j-1, source, target)
            val = 1 + min(insert_dist, delete_dist)
        # Change value in the matrix
        base.put(i, j, val)
    return base.get(i, j)
  \end{python}
  \end{algorithm}
  
  \begin{algorithm}[p]
  \caption{Adaptive algorithm to compute the \textsc{Delete-Insert Edit Distance} between two strings. 
    \label{algo:adaptiveDI}}
  \begin{python}
def adaptive_di(base, i, j, source, target):
    val = 0
    # Base cases
    if j < 0:
        return i+1
    elif i < 0:
        return j+1
    # Check if the cell has been filled     
    # default cell value is -1 
    if base.get(i, j) < 0:
        # Get ith and jth word froum source and target
        s = source.get(i)
        t = target.get(j)
        # Calculate rank in S[0:i] of t
        rank_s = source.rank(t, i)
        # Calculate rank in T[0:j] of t
        rank_t = target.rank(s, j)
        if s == t:
            val = adaptive_di(base, i-1, j-1, source, target)
        elif rank_s == 0 and rank_t == 0:
            val = 2 + adaptive_di(base, i-1, j-1, source, target)
        elif rank_s == 0:
            val = 1 + adaptive_di(base, i, j-1, source, target)
        elif rank_t == 0:
            val = 1 + adaptive_di(base, i-1, j, source, target)
        else:
            di_dist = 2 + adaptive_di(base, i-1, j-1, source, target)
            # Calculate select in S[0,rank_s] of t
            select_s = source.select(t, rank_s)
            # Calculate select in T[0,rank_t] of s
            select_t = target.select(s, rank_t)
            delete_dist = i - select_s + \
                adaptive_di(base, select_s-1, j-1, source, target)
            insert_dist = j - select_t + \
                adaptive_di(base, i-1, select_t-1, source, target)
            val = min(insert_dist, delete_dist, di_dist)
        # Change value in the matrix
        base.put(i, j, val)
    return base.get(i, j)
  \end{python}
  \end{algorithm}
  \subsection{DIR-Edit Distance}
  \label{sec:DIRExperiment}
  \begin{algorithm}[p]
  \caption{Classical algorithm to compute the \textsc{Delete-Insert-Replace Edit Distance} between two strings. 
    \label{algo:classicalDIR}}
  \begin{python}
def classic_dir(base, i, j, source, target):
    val = 0
    # Base cases
    if j < 0:
        return i+1
    elif i < 0:
        return j+1
    # Check if the cell has been filled
    # default cell value is -1
    if base.get(i, j) < 0:
        # Get ith and jth word froum source and target
        if source.get(i) == target.get(j):
            val = classic_dir(base, i-1, j-1, source, target)
        else:
            delete_dist = classic_dir(base, i-1, j, source, target)
            insert_dist = classic_dir(base, i, j-1, source, target)
            replace_dist = classic_dir(base, i-1, j-1, source, target)
            val = 1 + min(insert_dist, delete_dist, replace_dist)
        # Change value in the matrix
        base.put(i, j, val)
    return base.get(i, j)
  \end{python}
  \end{algorithm}
  
  \begin{algorithm}[p]
  \caption{Adaptive algorithm to compute the \textsc{Delete-Insert-Replace Edit Distance} between two strings. 
    \label{algo:adaptiveDIR}}
  \begin{python}
def adaptive_dir(base, i, j, source, target):
    val = 0
    # Base case
    if j < 0:
        return i+1
    elif i < 0:
        return j+1
    # Check if the cell has been filled
    # default cell value is -1
    if base.get(i, j) < 0:
        # Get ith and jth word froum source and target
        s = source.get(i)
        t = target.get(j)
        # Calculate rank in S[0:i] of t
        rank_s = source.rank(t, i)
        # Calculate rank in T[0:j] of t
        rank_t = target.rank(s, j)
        replace_dist = 1 + \
            adaptive_dir(base, i-1, j-1, source, target)
        if s == t:
            val = replace_dist-1
        elif rank_s == 0 and rank_t == 0:
            val = replace_dist
        elif rank_s == 0:
            # Calculate select in T[0,rank_t] of s
            select_t = target.select(s, rank_t)
            insert_dist = j - select_t +\
                adaptive_dir(base, i-1, select_t-1, source, target)
            val = min(replace_dist, insert_dist)
        elif rank_t == 0:
            # Calculate select in S[0,rank_s] of t
            select_s = source.select(t, rank_s)
            delete_dist = i - select_s +\
                adaptive_dir(base, select_s-1, j-1, source, target)
            val = min(replace_dist, delete_dist)
        else:
            # Calculate select in S[0,rank_s] of t
            select_s = source.select(t, rank_s)
            # Calculate select in T[0,rank_t] of s
            select_t = target.select(s, rank_t)
            delete_dist = i - select_s +
                adaptive_dir(base, select_s-1, j-1, source, target)
            insert_dist = j - select_t +\
                adaptive_dir(base, i-1, select_t-1, source, target)
            val = min(insert_dist, delete_dist, replace_dist)
        # Change value in the matrix
        base.put(i, j, val)
    return base.get(i, j)
  \end{python}
  \end{algorithm}

  \subsection{DR-Edit Distance}
  \label{sec:DRExperiment}
  \begin{algorithm}[p]
  \caption{Classical algorithm to compute the \textsc{Delete-Replace Edit Distance} between two strings. 
    \label{algo:classicalDR}}
  \begin{python}
def classic_dr(base, i, j, source, target):
    val = 0
    # Base cases
    if i < j:
        return source.len() + target.len() + 1
    if j < 0:
        return i+1
    # Check if the cell has been filled
    # default cell value is -1
    if base.get(i, j) < 0:
        # Get ith and jth word froum source and target
        if source.get(i) == target.get(j):
            val = classic_dr(base, i-1, j-1, source, target)
        else:
            delete_dist = classic_dr(base, i-1, j, source, target)
            replace_dist = classic_dr(base, i-1, j-1, source, target)
            val = 1 + min(delete_dist, replace_dist)
        # Change value in the matrix
        base.put(i, j, val)
    return base.get(i, j)
  \end{python}
  \end{algorithm}

  \begin{algorithm}[p]
  \caption{Adaptive algorithm to compute the \textsc{Delete-Replace Edit Distance} between two strings. 
    \label{algo:adaptiveDR}}
  \begin{python}
def adaptive_dr(base, i, j, source, target):
    val = 0
    # Base case
    if i < j:
        return source.len()+target.len()+1

    elif j < 0:
        return i+1
    # Check if the cell has been filled
    # default cell value is -1
    if base.get(i, j) < 0:
        # Get ith and jth word froum source and target
        s = source.get(i)
        t = target.get(j)
        # Calculate rank in S[0:i] of t
        rank_s = source.rank(t, i)
        replace_dist = 1 + adaptive_dr(
            base, i-1, j-1, source, target)
        if s == t:
            val = replace_dist-1
        elif rank_s == 0:
            val = replace_dist
        else:
            # Calculate select in S[0:rank_s] of t
            select_s = source.select(t, rank_s)
            delete_dist = i - select_s + \
                adaptive_dr(base, select_s-1, j-1, source, target)
            val = min(delete_dist, replace_dist)
        # Change value in the matrix
        base.put(i, j, val)
    return base.get(i, j)
\end{python}
\end{algorithm}

\end{PYTHON}

\end{document}